\documentclass{llncs}

\usepackage{algorithm}
 \usepackage{algpseudocode}
 \usepackage{amsmath,color}
 \usepackage{graphics}
 \usepackage{epsfig}

\newcommand{\Input}{\textbf{Input:\quad}}
\newcommand{\Initial}{\textbf{Initial:\quad}}
\newcommand{\Output}{\textbf{Output:\quad}}

\usepackage{graphicx}
\usepackage{epsf}
\usepackage{amssymb}
\usepackage[numbers,sort&compress]{natbib}

\newtheorem{thm}{Theorem}[subsection]
\newtheorem{prop}[thm]{Proposition}

\newtheorem{exam}[thm]{Example}
\newtheorem{ques}[thm]{Question}

\begin{document}

\title{A Linear Algebra Formulation for Boolean Satisfiability Testing\thanks{This work was partially supported by NSF of China (No. 11401061, No. 61202131 and No.61672448), CAS Youth Innovation Promotion Association (No. 2015315), Chongqing Science and Technology Commission projects cstc2014jcsfglyjs0005 and cstc2014zktjccxyyB0031.}}

\author{Chengling Fang\inst{1} \and Jiang~Liu\inst{2}} 
\institute{Chongqing Jiaotong University. \and
Chongqing Institute of Green and Intelligent Technology,\\
Chinese Academy of Sciences. \\
\email{liujiang@cigit.ac.cn.}}

\maketitle

\begin{abstract}
In the article ``The State of SAT'', the authors asked whether a procedure dramatically different from DPLL can be found for handling unsatisfiable instances. This study proposes a new linear programming approach to address this issue efficiently. Our experiments showed that the new method works for many unsatisfiable instances. However, we must concede that this method should be incomplete; otherwise, it will imply P=co-NP.
\end{abstract}


\section{Introduction}

The problem of determining whether a Boolean formula is unsatisfiable is called Boolean unsatisfiability (UNSAT) problem. Its opposite the Boolean satisfiability (SAT) problem is  famous  in mathematical logic and computing theory, which is one of the first proven NP-complete problems \cite{Cook1971, Levin1973}. SAT is widely studied because of its well known significance on  both of theory and practice \cite{SATSurvey96,cookPvsNP,Kautz2007Survey,Aaronson2016}.  Despite the worst-case exponential running time of all known algorithms for SAT, a lot of impressive progresses have been made in solving practical SAT problems with up to a million variables \cite{LargeVariablesApplication03,LargeVariableApplication2003}.

Based on the DPLL method \cite{DP1960, Davis1962}, there were developed a large number of high-performance algorithms for SAT: local search algorithms \cite{Selman95localsearch, Li03qingting, UnitWalk2005, StochasticLocalSearch2014}, stochastic algorithms \cite{PPSZ1998,   ProbabilisticAlgorithm1999, PPSZ2005}, conflict-driven clause learning algorithms  \cite{GRASP1999, chaff2001, CDLZhang2001}, and so on. These algorithms are somehow logic search methods. A second category of interesting  methods are related to constraint satisfaction problems which are used to employ various optimization strategies, for example, Lagrangian techniques \cite{Chang1995, Shang98adiscrete}, Newton's method and descent method for Universal SAT \cite{Gu1999}. A third typical method is based on statistical physics analysis which sugguests new effective heuristic algorithms for finding SAT assigments for random $k$-SAT problems \cite{yedidia2001generalized, maneva2007a}.
There are too many elegant work to mention, we apologize for missing references, and for more related work please refer to some survey papers such as \cite{Gu1999,Kautz2007Survey,Aaronson2016} and  the references therein.  Many of them also can be used to study UNSAT indirectly.

Comparing with the large number of studies about SAT, there were a few direct work \cite{GUNSAT2007} about UNSAT. Note that, a formula being unsatisfiable is logically equivalent to its negation being valid. So UNSAT amounts to the tautology problem which is co-NP-complete.  When an unsatisfiable conjunctive normal form (CNF) formula contains too many clauses, its searching space is intractably huge. Naturally, any DPLL based method for such unsatisfiable formula will require a huge time. In practice, it was ever for a long time that there was no local search algorithm for UNSAT before GUNSAT \cite{GUNSAT2007} was proposed.  Therefore it \cite{Kautz2007Survey} was eagerly concerned whether a procedure dramatically different from DPLL can be found for handing UNSAT. This study presents a novel method called linear algebra formulation (LAF) to address this issue efficiently.


In the study of SAT, researchers are used to restrict SAT to a special categories such as $k$-SAT, XOR-SAT, Horn-SAT,  1-in-3-SAT and so on. By the Schaefer's dichotomy theorem \cite{Schaefer1978}, each restriction is either in $\mathbf{P}$ or $\mathbf{NP}$-complete. In above categories, the 1-in-3-SAT received our attention because we can establish a natural relation between it and a system of linear equations. A $3$-CNF formula  is called 1-in-3 satisfiable if there is a truth assignment to its Boolean variables such that each clause has exactly one true literal, otherwise 1-in-3 unsatisfiable. The 1-in-3-SAT problem is to determine whether a $3$-CNF formula is 1-in-3 satisfiable, which is $\mathbf{NP}$-complete \cite{Schaefer1978}. Similarly,   the 1-in-3-UNSAT problem is  to determine whether a $3$-CNF formula is 1-in-3 unsatisfiable.

The basic idea of LAF for UNSAT is as follows. It first converts the UNSAT problem into a 1-in-3-UNSAT problem. Then it converts the 1-in-3-UNSAT problem into a Boolean solution (BoS) problem of the corresponding system of linear equations, where a BoS is a solution composed merely of $0$ and $1$. For the resulted linear system, we develop a linear algebra formulation to efficiently test whether it has any BoS. Through this approach, we obtain some sufficient conditions for UNSAT. Let's explain the idea by the following toy example.

\begin{exam}
Consider the following  $3$-CNF formula
	\begin{eqnarray}
	(X_1\vee X_2\vee X_3) \wedge (X_2\vee X_3\vee X_4) \wedge (X_1\vee X_4) \label{fm:unsatexam}
	\end{eqnarray}
	which is   1-in-3 unsatisfiable. First, Formula (\ref{fm:unsatexam}) is transformed into a linear system
	\begin{eqnarray}
	X_1+X_2+X_3 & = & 1    \nonumber\\
	X_2+X_3+X_4 & = & 1  \label{eq:unslb}  \\
	X_1+X_4 & = & 1     \nonumber
	\end{eqnarray}
	with an abuse of notations of Boolean variable and equation variable by the same symbol. As a result, formula (\ref{fm:unsatexam}) is 1-in-3-SAT iff system (\ref{eq:unslb}) has some BoS. To restrict  $X_1,\cdots, X_4$ being $0$ or $1$, it just puts the following quadratic constraints
	\begin{eqnarray}\label{eq:qdct}
	X_1 = X_1^2, X_2 = X_2^2, X_3 = X_3^2, X_4 & = & X_4^2   \label{eq:unslbs}
	\end{eqnarray}
	Now, formula (\ref{fm:unsatexam}) is 1-in-3-SAT iff the polynomial system consisting of (\ref{eq:unslb}) and (\ref{eq:unslbs}) has a real solution.  Unfortunately, so far there is no fast (polynomial time) algorithm for deciding if a  quadratic system has a real solution.
	
	We then turn to a quadratic system which has the same real solution and contains merely two degree monomials as follows
	\begin{eqnarray}\label{exam:inerprod}
	X_1^2+X_2^2+X_3^2 & = & 1   \label{eq:cd} \\
	X_2^2+X_3^2+X_4^2 & = & 1   \\
	X_1^2+X_4^2 & = & 1    \label{eq:sqeqn} \\
	(X_1+X_2+X_3)^2 & = & 1    \\
	(X_2+X_3+X_4)^2 & = & 1   \\
	(X_1+X_4)^2 & = & 1    \\
	(X_1+X_2+X_3)\cdot(X_2+X_3+X_4) & = & 1    \label{eq:intps}\\
	(X_1+X_2+X_3)\cdot(X_1+X_4) & = & 1   \\
	(X_2+X_3+X_4)\cdot(X_1+X_4) & = & 1   \label{eq:intpd} \\
	X_1\cdot(X_2+X_3+X_4) & = & X_1^2   \label{eq:intprod}  \\
	X_2\cdot(X_1+X_4) & = & X_2^2  \\
	X_3\cdot(X_1+X_4) & = & X_3^2    \\
	X_4\cdot(X_1+X_2+X_3) & = & X_4^2    \label{eq:szd} \\
	X_1\cdot(X_1+X_2+X_3) & = & X_1^2  \label{eq:zqm}    \\
	X_2\cdot(X_1+X_2+X_3) & = & X_2^2   \\
	X_3\cdot(X_1+X_2+X_3) & = & X_3^2   \\
	X_2\cdot(X_2+X_3+X_4) & = & X_2^2   \\
	X_3\cdot(X_2+X_3+X_4) & = & X_3^2   \\
	X_4\cdot(X_2+X_3+X_4) & = & X_4^2   \\
	X_1\cdot(X_1+X_4) & = & X_1^2    \\
	X_4\cdot(X_1+X_4) & = & X_4^2     \label{eq:zud}
	\end{eqnarray}
	Similarly, system (\ref{eq:unslb}) has a BoS iff the quadratic system consisting of euqations (\ref{eq:cd}-\ref{eq:zud}) has a BoS.
	
	Now, we relinearize  system (\ref{eq:cd}-\ref{eq:zud}) by substituting all monomials $X_iX_j$ and $X_jX_i$ by a single variable with $i\leq j$. Solving the linearization system obtains
	\begin{eqnarray}
	X_1X_2=X_1X_3 =X_1X_4=X_2X_3=X_2X_4 =X_3X_4= 0  \label{exam:cdxz}
	\end{eqnarray}
	From equations  (\ref{eq:intprod}-\ref{eq:szd}) and (\ref{exam:cdxz}), it must be
	\begin{eqnarray}\label{exam:sqz}
	X_1^2 =X_2^2=X_3^2=X_4^2= 0
	\end{eqnarray}
	It is evidently contradicting with the system composed by equations (\ref{eq:cd}-\ref{eq:sqeqn}). That is, the relinearized system  of the quadratic system (\ref{eq:cd})-(\ref{eq:zud}) is inconsistent. Thus, the quadratic system  (\ref{eq:cd})-(\ref{eq:zud}) has no real solution. As a result, system (\ref{eq:unslb}) has no BoS. Therefore, formula (\ref{fm:unsatexam}) is 1-in-3 unsatisfiable.
\end{exam}

The above example shows that the inconsistency of the relinearized system can be utilized to infer 1-in-3-UNSAT of a $3$-CNF formula. In the following, we formalize and generalize the idea and techniques in above toy example.  For the sake of clarity, through the article we take the notations roughly as following: capital letters with subindex are used for Boolean variables and variables for equations; script letters $\ell $, $\mathcal C$ and $\mathcal F$, etc, stand for literals, clauses and formulas;  low case letters $f, g, h$ and capital Greek letters $\Delta, \Theta$, etc, are used for Boolean or real functions; $\mathbb B=\{ 0,   1\}$ be the set of Boolean values, accordingly  $\mathbb B^n$ is the $n$-dimensional Boolean space; $\mathbb R$/$\mathbb N$ is the set of real/integer numbers.



\section{Formal Linear Algebra Formulation}\label{sec:basic}
Because Boolean unsatisfiability  problem can be efficiently reduction to 1-in-3-UNSAT which is co-NP-complete, it suffices to study efficient method for 1-in-3-UNSAT. To resolve 1-in-3-SAT/UNSAT, our basic idea is to convert it into  consistency problems  of the related linear system, over various fields. It consists of several crucial processes as follows
\begin{enumerate}
\item First one is the linear transformation (LT) that converts a Boolean formula into a linear system such that the Boolean formula is 1-in-3 satisfiable iff the linear system has some BoS;
\item Second one is the quadratic propagation (QP) that extends the transformed linear system into some quadratic system such that these two systems have the same BoSs;
\item Third one is the relinearization (ReL) that abstracts the quadratic system as a linear
system such that they have the same BoSs.
\end{enumerate}

In above procedure, it contains linearization twice. One is in the conversion from a Boolean formula into a linear system. Second one is abstracting a quadratic system as a linear system. So this approach is called {\it linear algebra formulation} (LAF) to highlight the status of linearization.

In the following, the previous procedure will be formulated formally.  To this end, we follow the standard concepts of propositional logic formula in terms of {\it literals} and {\it lauses}, conjunctive normal form (CNF). Given $n$-many Boolean variables $X=\{X_1, \cdots, X_n\}$, a  CNF formula $\mathcal F$ is defined as
\begin{eqnarray}\label{fm:CNF}
\mathcal F & = & \wedge_{i\leq m} \mathcal C_i   =  \wedge_{i\leq m} \vee_{j\leq j_i}\ell_{ij}(X)   \label{fm:CNFC}
\end{eqnarray}
where $\mathcal C_i=  \vee_{j\leq j_i}\ell_{ij}(X)$ are clauses, and each literal $\ell_{ij}(X)$ is of form $X_u$ or $\neg X_u$ for some $1\leq u\leq n$. As convention,  $\mathcal F$ has {\it pure  polarity} if at most one of $X_u$ and $\neg X_u$ can occur in $\mathcal F$  for any  $1\leq u\leq n$;  $\mathcal F$ is a {\it positive formula}  if merely $X_u$ can occurs in  $\mathcal F$; it  is a {\it $k$-CNF formula\footnote{In this article, a $k$-CNF formula  is a CNF formula in which each clause has at most $k$ literals. In some other literatures, a $k$-CNF formula is a CNF formula in which each clause has exactly $k$ literals.}} if $j_i\leq k$ for all $i\leq m$.

The notion of  1-in-3-SAT/UNSAT is crucial to LT, which will be extended to general CNF formula. A CNF formula $\mathcal F$ of form (\ref{fm:CNF}) is called {\it exactly one satisfiable} (EOS) if there is a truth assignment $X^*$ to the variables $X$ such that each clause $\mathcal C_i$ has exactly one true literal, otherwise called {exactly one unsatisfiable} (EOU).

For a given positive formula $\mathcal F$ of form (\ref{fm:CNF}), it can make a reduction of its EOS to the existence of BoS of a linear system defined by
\begin{eqnarray}\label{eq:ATLBS}
\sum_{j\leq j_i} \ell_{i,j}(X) & = & 1 \quad \quad \quad 1\leq i\leq m, \label{eq:TLBS}
\end{eqnarray}
where all $\ell_{i,j}(X)=X_k$ become  equation variables from  positive literals in $\mathcal F$. Let $\textbf{True}\leftrightarrow 1, \textbf{False}\leftrightarrow 0$ be the one-to-one relation between truth values $\{\textbf{True}, \textbf{False}\}$ and Boolean values $\{1, 0\}$, then
\begin{prop}\label{prop:BoS-EOS}
A positive CNF formula $\mathcal F$ of form (\ref{fm:CNF}) is EOS iff the corresponding linear system (\ref{eq:ATLBS}) has a BoS.
\end{prop}

As a result, we can study the EOS of a positive formula through investigating the BoS of the corresponding linear system.

 Two formulas $\mathcal F$ and $\mathcal H$ are said {\it equi-exactly-one-satisfiable} (equi-EOS) if $\mathcal F$ is exactly one satisfiable whenever $\mathcal H$ is  and vice versa. If the $\mathcal F$ defined by (\ref{fm:CNF}) is of pure polarity, then we could construct an equi-EOS  positive formula $\mathcal F^*$ by simply substituting all negation literal $\neg X_u$ with $X_u$. As $X_u$ does not occur in $\mathcal F$,  $\mathcal F$ and $\mathcal H$ must be equi-EOS. Therefore, we can conclude that
\begin{prop}\label{prop:PP2P}
Each pure polarity CNF formula is equi-EOS to a positive CNF formula.
\end{prop}

Therefore, the EOS of each pure polarity formula can also be studied through some linear system. In general, a CNF formula $\mathcal F$ may have no pure polarity property. In this case, we introduce auxiliary variables $Y_u$ for all $\neg X_u$. Then we could construct an equi-EOS positive formula $\mathcal F^*$ as follows.  Let $\mathcal F(Y/\neg X)$ be the resulted formula by substitute all negative literals $\neg X_u$ with $Y_u$ in $\mathcal F$. Let
\begin{eqnarray}\label{eq:G2PF}
\mathcal F^* & = &  \wedge_u (X_u\vee Y_u) \wedge  \mathcal \mathcal F(Y/\neg X)  \label{eq:GtoPF}
\end{eqnarray}
where $u$ ranges in the index set $I$ such that if $\neg X_u$ occurs in  $\mathcal F$ then $u\in I$. Such $\mathcal F^*$ is called a {\it positivization} of $\mathcal F$, also denoted by $\mathcal P(\mathcal F)$. Now
\begin{prop}\label{prop:G2P}
Each CNF formula $\mathcal F$  is equi-EOS to a positive formula   $\mathcal P(\mathcal F)$.
\end{prop}

Accordingly, it can convert the EOS  of  $\mathcal F$ into the existence of BoS of the following linear system
\begin{eqnarray}\label{eq:GATLBS}
\sum_{j\leq j_i} \ell_{i,j}^*(X) & = & 1, \quad \quad \quad 1\leq i\leq m \label{eq:GLBS}  \label{eq:PL}   \\
X_u+Y_u & = & 1 \quad \quad \quad  u\in I   \label{eq:PLC}
\end{eqnarray}
where $\ell_{i,j}^*(X)=X_u$ if $\ell_{i,j}(X)=X_u$ and $\ell_{i,j}^*(X)=Y_u$ if $\ell_{i,j}(X)=\neg X_u$. As a consequence, it can convert the EOS problem of a generic CNF formula $\mathcal F$  into the existence of BoS of the related linear system. The transformation from a CNF formula $\mathcal F$ into a linear system like (\ref{eq:ATLBS}) or (\ref{eq:GATLBS}) is the so-called {\it linearizing transformation}. Without loss of  generality, we study the BoS of linear system (\ref{eq:ATLBS}) for positive formula in stead of generic ones in the sequel.

Generally, it is $\mathbf{NP}$-hard to decide whether a linear system (\ref{eq:GATLBS}) has a BoS.  Anyway, we could exploit the idea behind the toy example in the previous section to approximate BoS. The basic idea is to resort to some easily solving BoS-equisolvable linear system, where two linear systems $\mathcal L1$ and $\mathcal L2$ are BoS-equisolvable if they satisfy: $\mathcal L1$ has  BoS iff $\mathcal L2$ has BoS. In light of this, we extend a given linear system $\mathcal L$ to some  BoS-equisolvable  linear system $\mathcal L^*$ containing $\mathcal L$ as a subsystem. To this end, the equations in $\mathcal L^*$ are obtained by two consecutive algebraic operations on $\mathcal L$, which are  QP and ReL. Herein, the QP over (\ref{eq:ATLBS}) is consisted of two sorts of operations. One is done by mutually multiplying equations inside $\mathcal L$ side by side, which is called {\it inner quadratic propagation} (IQP),  another is accomplished through side by side  multiplications over equations of $\mathcal L$ and the following quadratic constraints
\begin{eqnarray}\label{eq:QCT}
X_u & = & X_u^2, \quad \quad \quad 1\leq u\leq n \label{eq:QC}
\end{eqnarray}
which is called {\it constraint quadratic propagation} (CQP). Formally, IQP and  CQP are carried out respectively by
\begin{eqnarray}
(\sum_{j\leq j_i} \ell_{i,j}(X))\cdot (\sum_{j\leq j_t} \ell_{t,j}(X)) & =  1, &  1\leq i\leq t\leq m  \label{eq:QPF}  \label{eq:QP}  \label{eq:IQP} \\
X_u \cdot  \sum_{j\leq j_i} \ell_{i,j}(X) & =  X_u^2, &\quad  1\leq u\leq n, \quad  1\leq i\leq m   \label{eq:CQP}
\end{eqnarray}


ReL is to substitute all quadratic monomials $X_iX_j$ in the quadratic system  (\ref{eq:QP}) with new variables, $Z=(Z_1, \cdots, Z_u, \cdots, Z_v)$ says, in order to transform such quadratic system into a linear system. Wherein, $u\leq v= \frac{n(n+1)}{2}$. Let $ReL(Z)$ denote such linear system, then for a positive $k$-CNF formula
\begin{thm}\label{thm:main}
The following three assertions are equivalent:
\begin{enumerate}
\item The linear system (\ref{eq:ATLBS}) has BoS;
\item The quadratic system (\ref{eq:QP}-\ref{eq:CQP}) has BoS;
\item The $ReL(Z)$  has BoS.
\end{enumerate}
\end{thm}
\begin{proof}
It is easy to show implications from 1) to 2), and from 2) to 3). In fact,  $ReL(Z)$  can be simplified so that it contains (\ref{eq:ATLBS}) as a subsystem. Therefore, 3) naturally implies 1).
\end{proof}

Based on this theorem and Proposition \ref{prop:BoS-EOS}, we obtain a sufficient condition for that a positive formula $\mathcal F$ has no BoS, as follows.
\begin{prop}\label{prop:sound}
If $ReL(Z)$ has no solution over anyone of $\mathbb R^v$ and $\mathbb N^v$, then it definitely has no BoS. Therefore, $\mathcal F$ is EOU.
\end{prop}

As a result, the 1-in-3-SAT/UNSAT problem can be converted into a consistency problem of certain linear system over some field or ring. Sometime, it can obtain a satisfying assignment for certain  1-in-3-SAT instance.

\begin{exam}\label{exam:unqsat}
	Consider the following $3$-CNF formula
	\begin{eqnarray}
	(X_1\vee X_2 \vee X_3)\wedge (X_2\vee X_4 \vee X_5) \wedge (X_2\vee X_6)  \wedge (X_3\vee X_4 \vee X_6)  \label{fm:toySAT}
	\end{eqnarray}
Applying LAF, we get it's ReL
	\begin{eqnarray}
	Y_{11}+Y_{22}+Y_{33} & = & 1     \nonumber\\
	& \vdots &   \nonumber \\
	Y_{33}+Y_{44}+Y_{66} & = & 1     \label{eq:rlq}\\
	& \vdots &   \nonumber \\
	Y_{36}+Y_{46}+Y_{66} & = & Y_{66}    \nonumber
	\end{eqnarray}
	It is easy to verify that system (\ref{eq:rlq}) has only one solution
	\begin{eqnarray}
	Y_{11}=Y_{15}=Y_{16}=Y_{55}=Y_{56}=Y_{66}=1      \label{eq:BoSu}\\
	Y_{12}=Y_{13}=Y_{14}=Y_{22}=Y_{23}=Y_{24}=Y_{25}=Y_{26}=\cdots=Y_{45}=0    \label{eq:BoSz}
	\end{eqnarray}
	Accordingly, $(X_1,X_2,X_3,X_4,X_5,X_6)=(1,0,0,0,1,1)$ must be the unique BoS to the linear system  associated with (\ref{fm:toySAT}), and correspondingly, $(T,F,F,F,T,T)$ is the unique truth assignment by which formula (\ref{fm:toySAT}) is 1-in-3 satisfiable.
	
\end{exam}

\section{Algorithms and Experiments}\label{sec:alg}
In the previous section, the LAF for  EOS/EOU was established on mathematically rigorous foundation.  It is the core of LAF for general SAT. Hence, we present it as  Algorithm \ref{alg:LAT4EOS}. Herein, outcome `EOS'  means that $\mathcal F$ is exactly one satisfiable; outcome `EOU' means that $\mathcal F$ is not so; outcome `Unk' stands for  that the answer is unknown by the method. The soundness of Algorithm \ref{alg:LAT4EOS} is guaranteed by Propositions \ref{prop:BoS-EOS} and \ref{prop:sound}.

\begin{algorithm}[t]
\scriptsize   \caption{Kernel Algorithm: LAT for EOU}   \label{alg:LAT4EOS}
\Input   A 3-CNF  formula $\mathcal F(X) = \mathcal C_i(X) =\wedge_i \vee_j L_{ij} (X)$ with
variables $X = \{X_1,\cdots, X_n\}$;   \\
\Initial Answer=  `Unknown';  \\
\Output     Answer=`EOS', or `EOU', or `Unk'.

\begin{algorithmic}[1]
\State Do LT on  $\mathcal F(X)$ and obtain a linear system (\ref{eq:ATLBS})
\label{code:FL-LBS}
\State Decide whether (\ref{eq:ATLBS}) is consistent over $\mathbb R^n$ and $\mathbb N^n$
 \If {it is not consistent over anyone of them}
  \State  Set Answer=`EOU', return
  \If {it is consistent}
   \State check if (\ref{eq:ATLBS}) has a unique solution over any of $\mathbb R^n$ and $\mathbb N^n$ \label{code:US}
    \If {(\ref{eq:ATLBS}) has a unique solution, $X^*$ says,}
      \State Decide whether  $X^*$ is a BoS
       \If {$X^*$ is a BoS}
        \State Set Answer= `EOS', return
         \If {$X^*$ is not a BoS}
         \State Set Answer=`EOU', return
         \If {(\ref{eq:ATLBS}) always has more than one solution}
    \State Go to \ref{code:QP}
\EndIf
\EndIf
\EndIf
\EndIf
\EndIf
\EndIf

\State Do QP on (\ref{eq:ATLBS}), and obtain a quadratic system (\ref{eq:QP})  \label{code:QP}
\State Do relinearization by substitution on (\ref{eq:QP}), and obtain a linear system, says $ReL(Z)$    \label{code:LA}
\State Decide whether $ReL(Z)$ is consistency over $\mathbb R^v$ and $\mathbb N^v$  \label{code:lz}
   \If {$ReL(Z)$  is inconsistency over anyone of them}
    \State Set Answer=`EOU', return
       \If {$ReL(Z)$ is always consistency}
       \State Set Answer= `Unk',  return
 \EndIf
\EndIf   \\
\Return Answer  \label{code:end}
\end{algorithmic}
\end{algorithm}

Table \ref{tab:eos} reports the experiment results using Algorithm \ref{alg:LAT4EOS}. In the experiments, the instances are  randomly generated $3$-CNF formulas. In Table \ref{tab:eos}, {\bfseries  $\#$T} stands for the number of instances with {\bfseries $\#$V} many variables and {\bfseries  $\#$C} many clauses; {\bfseries  $\#$Unk}, {\bfseries  $\#$EOS} and {\bfseries $\#$EOU}  denote the numbers of corresponding answers. The experiment results confirmed that LAF make essential significance for EOU. However, the results showed that LAF is not good for EOS instances. Anyway, it also provides us some insights of 1-in-3-SAT. For example, a 3-CNF formula often is  1-in-3-UNSAT when the number of its clauses is more than 90 percent of the number of its variables.

\begin{table}[!h]
\caption{Implementation of Kernel Algorithm}
\label{tab:eos}
\centering
\begin{tabular}{c|c|c|c|c|c}
\hline
~~\bfseries $\#$V~~ &~~ \bfseries  $\#$C~~ &~~ \bfseries  $\#$T ~~ &  \bfseries  $\#$Unk  & \bfseries  $\#$EOS &\bfseries $\#$EOU      \\
\hline
50 & 41 & 100 & 12 & 9 & 79 \\
\hline
50 & 46 & 100 & 0 & 0 & 100 \\
\hline
70 & 58 & 100 & 8 & 0 & 92 \\
\hline
70 & 66 & 100 & 0 & 1 & 99 \\
\hline
90 & 74 & 100 & 11 & 0 & 89 \\
\hline
90 & 82 & 100 & 0 & 0 & 100 \\
\hline
130 & 109 & 100 & 15 & 0 & 85 \\
\hline
130 & 118 & 100 & 0 & 0 & 100 \\
\hline
150 & 125 & 100 & 36 & 0 & 64 \\
\hline
150 & 136 & 100 & 0 & 0 & 100 \\
\hline

\end{tabular}
\end{table}

As 1-in-3-SAT is $\mathbf{NP}$-complete \cite{Schaefer1978}, SAT can be reducible to  1-in-3-SAT efficiently in polynomial time. Therefore, LAF can be used for efficiently resolving general SAT/UNSAT. The recipe  is to perform a series of equisatisfiable transformations as follows. Given a general formula $\mathcal G(X)$ with variables $X=\{X_1,\cdots, X_n\}$, we carry out the following process:
\begin{enumerate}
\item First, we transform $\mathcal G(X)$ into an equisatisfiable CNF formula, says $\mathcal G^*(X,Y)$.
\item For $\mathcal G^*(X,Y)$, it computes an equisatisfiable 3-CNF formula $\mathcal T(X,Y,Z)$.
\item Based on $\mathcal T(X,Y,Z)$,  a positive formula $\mathcal F(X,Y,Z,U)$ is computed such that $\mathcal T(X,Y,Z)$ is satisfiable if and only if $\mathcal F(X,Y,Z,U)$ is 1-in-3 satisfiable.
\item Algorithm \ref{alg:LAT4EOS} is applied to $\mathcal F(X,Y,Z,U)$. If Algorithm \ref{alg:LAT4EOS} outputs ``EOS'' for $\mathcal F(X,Y,Z,U)$ then $\mathcal G(X)$ is satisfiable, similarly, if ``EOU'' is output then $\mathcal G(X)$ is unsatisfiable. Otherwise, the answer is ``Unk'',  that is, satisfiability of $\mathcal G(X)$ is not decided by LAF.
\end{enumerate}
The whole procedure could be formally summarized by Algorithm \ref{alg:LAT4SAT}.
\begin{algorithm}[t] \scriptsize  \caption{LAF Test for SAT}   \label{alg:LAT4SAT}
\Input   A CNF formula $\mathcal G(X)$ with variables $X$;   \\
\Initial Answer=  `Unknown';  \\
\Output     Answer=`SAT', or `UNSAT', or `Unk'.
\begin{algorithmic}[1]
\State Transform  $\mathcal G(X)$ into an equisatisfiable CNF formula $\mathcal G^*$
\label{code:CNFTrans}
\State Compute an equisatisfiable 3-CNF formula $\mathcal T$ for $\mathcal G^*$
\label{code:3-SATTrans}
\State Compute a 3-CNF positive formula $\mathcal F$ for $\mathcal T$
\label{code:PFTrans}
\State Call Algorithm \ref{alg:LAT4EOS} for $\mathcal F$
\If {Output is `EOS'}
\State Set Answer=`SAT', return
\If {Output is `EOU'}
\State Set Answer=`UNSAT', return
\If {Output is `Unk'}
\State Set Answer=`Unk', return
\EndIf
\EndIf
\EndIf   \\
\Return Answer  \label{code:satend}
\end{algorithmic}
\end{algorithm}

Another major concern of an algorithm is its computational complexity. A short complexity analysis of these two algorithms is performed in what follows. For Algorithm \ref{alg:LAT4EOS}, its complexity is mainly due to the decision of consistency of a linear system and the implementation of QP. Given a 3-CNF positive formula $\mathcal F$ of $m$-clauses and $n$-variables with $m\leq n$, in the final linear system $ReL(Z)$, there are $\frac{n(n+1)}{2}$ many variables and $mn+\frac{m(m+1)}{2}+n$ many linear equations. For consistency of linear systems, there is an algorithm \cite{Kaltofen2015} of complexity $O(MNr)$ to decide whether a linear system of $M$ linear equations and $N$ variables is consistent over $\mathbb R$, where $r$ is the rank of the coefficient matrix. When we use the consistency over $\mathbb N$, it needs to compute a full row rank form coefficient matrix to compute Hermit normal form (HNF) \cite{Schrijver1998}. For HNF, the algorithm in \cite{Micciancio2001} is capable to convert an integer $s\times t$ matrix into HNF with complexity $O(st^4)$, where integers $s\leq t$.   Therefore, the complexity of Algorithm \ref{alg:LAT4EOS} is about $O(n^{10})$. Similarly, Algorithm \ref{alg:LAT4SAT} terminates in polynomial time since its additional actions for converting a general CNF formula into a 3-CNF positive formula is of polynomial size of the numbers of variables and clauses.

Interestingly, if the inconsistency of $ReL(Z)$ over real number field is also a necessary condition for that the corresponding 3-CNF positive formula is not 1-in-3 satisfiable, then it can modify Algorithm \ref{alg:LAT4EOS} by substituting  `Unk' with `EOS'.
Accordingly, 1-in-3-SAT would be decided in time $O(n^4m^2+m^4n^2+m^3n^3)$. As a result, SAT would be solved in polynomial time. Unfortunately, the inconsistency of $ReL(Z)$ over $\mathbb R$ cannot be a such necessary condition,   here is a counterexample
\begin{eqnarray}
\mathcal F & = &(X_1 \vee X_2 \vee X_3) \wedge (X_2\vee X_4\vee X_5) \wedge (X_2\vee X_6) \wedge (X_3\vee X_4 \vee X_6) \nonumber \\
 & & \wedge (X_1 \vee X_7 \vee X_8) \wedge (X_1\vee X_{9}\vee X_{10}) \wedge (X_1\vee X_{11}\vee X_{12}) \nonumber \\
  & & \wedge (X_7 \vee X_{13} \vee X_{14}) \wedge (X_{9}\vee X_{13} \vee X_{15}) \wedge (X_{11}\vee X_{14} \vee X_{15})
\end{eqnarray}
However, its corresponding $ReL(Z)$ has no solution over $\mathbb N$, which still can show that $\mathcal F$ is 1-in-3-UNSAT. Nevertheless, in our experiments there are several such cases. Therefore, it is interesting to ask
\begin{ques}\label{question}
Given a 3-CNF formula $\mathcal F$, whether the inconsistency of $ReL(Z)$  over $\mathbb N$ is a sufficient condition for it being 1-in-3-UNSAT?
\end{ques}

If the answer  is `yes', then it can modify  Algorithm \ref{alg:LAT4SAT} accordingly and  obtain  a definite answer `SAT' or ``UNSAT' for each formula input. In such case, the SAT and UNSAT both can be solved in polynomial time. This will lead to P={NP}=co-NP. Anyway,  $\mathbf P\neq \mathbf{NP}$ is a overwhelming opinion \cite{cookPvsNP,Aaronson2016}  currently. So the most possible answer might be `no'. In this case, it is natural to ask what is the class of Boolean formulas whose satisfiability can be determined by the inconsistency of $ReL(Z)$  over $\mathbb N$?

\section{Summary}\label{sec:con}

This study proposes a novel method LAF to SAT/UNSAT. This method mainly establishes an equivalent relation between satisfiability of Boolean formulas and Boolean solvability of linear system, and brings up a new approach for Boolean solution to linear system. As can be seen, LAF is a procedure dramatically different from DPLL. Hence, it gave an affirmative answer to the question in the end of \textbf{Challenge 1} in \cite{Kautz2007Survey}.  More importantly, we developed two polynomial time algorithms for unsatisfiability testing based upon LAF. However, it can only say that LAF is an incomplete method for SAT unless Question \ref{question}  has a affirmative answer.  Nevertheless, LAF has been employed to successfully prove 1-in-3-UNSAT for many nontrivial cases in the experiment. So far, LAF is mainly used to show EOU especially 1-in-3-UNSAT. In addition to Question \ref{question}, it is also interesting to study how to develop LAF to compute a satisfying assignment for satisfiable formulas.

\bibliographystyle{splncs03}

\end{document}